\documentclass{llncs}
\usepackage[utf8]{inputenc}

\usepackage{adjustbox}
\usepackage{amsmath} % \align
\usepackage{amssymb} % \checkmark
\usepackage{complexity} % \NP
\usepackage{multirow}
\usepackage[numbers,sectionbib]{natbib}

\usepackage{tikz}
\usetikzlibrary{arrows}
\usetikzlibrary{shapes.geometric} 

\tikzstyle{vertex} = [circle,fill=black!0,minimum size=3pt,inner sep=0pt]
\definecolor{Blueish}{rgb} {0,0,0.55}
\definecolor{Redish}{rgb} {0.95,0.4,0.3}
\definecolor{Greenish}{rgb} {0.6,1,0.6}

\tikzstyle{part} = [line width=4pt,opacity=.5,cap=round,join=round,to path={(\tikztostart.center)--(\tikztotarget.center)\tikztonodes}]

\newcommand{\netgraph}{
    \node[draw,minimum size=.75cm,regular polygon,regular polygon sides=3] (triangle) {};      
    \foreach \x/\lab in {1/b,2/c,3/d} 
        \node[vertex,draw] at (triangle.corner \x) (\x) {};  

    \foreach [count=\i from 1] \angle/\name in {90/11,225/22,-45/33}
        \draw (\i) -- ++(\angle:.5cm) node [vertex,draw] (\name) {};
}

\newcommand{\diamondgraph}[1][0]{
    \node[draw,minimum size=.75cm,regular polygon,regular polygon sides=4,rotate=#1] (square) {};    
    \draw (square.corner 1) -- (square.corner 3);
    \foreach \x/\lab in {1/b,2/c,3/d,4/a} 
        \node[vertex,draw] at (square.corner \x) (\lab) {};    
}

\newcommand{\diamondbeforepartition}{
    
    \begin{scope}
        \diamondgraph[45];
    \end{scope}
    
    \node at (b) [label=above:{$v$}]{};
    \node at (c) [label=above:{$u$}]{};
    \node at (d) [label=below:{$w$}]{};
    \node at (a) [label=above:{$x$}]{};

     \draw (c) -- ++(180:1.5cm) node [vertex,draw,label=left:$u'$] (5) {};
     \draw (a) -- ++(0:1.5cm) node [vertex,draw,label=right:$x'$] (6) {};

    \draw[gray,very thick,line width=2pt,opacity=.3] (5) -- ++(120:.5cm)  node {};
    \draw[gray,very thick,line width=2pt,opacity=.3] (5) -- ++(240:.5cm)  node  {};

    \draw[gray,very thick,line width=2pt,opacity=.3] (6) -- ++(60:.5cm)  node {};
    \draw[gray,very thick,line width=2pt,opacity=.3] (6) -- ++(-60:.5cm)  node  {};
}

\newcommand{\cofish}{
    \foreach \angle/\name in {0/a,90/b,180/c,270/d}
        \path (\angle:1cm) node [vertex,draw] (\name) {};
    
    \draw (a) -- (b) -- (c) -- (d) -- (a);
    \draw (b) -- node[draw,vertex] (e) {} (d);
    \draw (c) -- (e);
    \draw (a) -- ++(0:1cm) node[vertex,draw] (f) {};
}

\newcommand{\trianglefreelobe}{
    \draw (0,0) circle (1cm);
    \path (90:1cm) node [vertex,draw] (a) {};
    \path (135:1cm) node [vertex,draw] (b) {};
    \path (-135:1cm) node [vertex,draw] (c) {};
    \path (-45:1cm) node [vertex,draw] (d) {};
    \path ( 45:1cm) node [vertex,draw] (e) {};

    \draw (0,0) circle (.44cm);
    \foreach [count=\i from 1] \angle in {45,135,-45,-135}
        \path (\angle:.44cm) node [vertex,draw] (\i) {};

    \draw (1) -- (e);
    \draw (3) -- (d);
    \draw (2) -- (b);
    \draw (4) -- (c);
    \draw (a) -- ++(90:1cm) node[vertex,draw] (f) {};
}

\newcommand{\varandclausegadget}{
    
    \node[draw,minimum size=.75cm,regular polygon,regular polygon sides=5] (C5below) {};

    \draw[densely dotted] (C5below.corner 1) -- ++(90:.5cm)  node[vertex,draw,label=above:$x_i$,solid] (u) {};
    \draw[densely dotted] (u) -- ++(150:.5cm)  node[vertex,draw,solid] (v){};
    \draw[densely dotted] (u) -- ++(30:.5cm)  node[vertex,draw,solid] (w){};
    \draw[gray,very thick,line width=2pt,opacity=.3] (v) -- ++(45:.25cm)  node {};
    \draw[gray,very thick,line width=2pt,opacity=.3] (v) -- ++(135:.25cm)  node {};
    \draw[gray,very thick,line width=2pt,opacity=.3] (w) -- ++(45:.25cm)  node {};
    \draw[gray,very thick,line width=2pt,opacity=.3] (w) -- ++(135:.25cm)  node {};

    \draw[densely dotted] (C5below.corner 3) -- ++(225:.5cm)  node[vertex,draw,label=south west:$x_j$,solid] (u1) {};
    \draw[densely dotted] (u1) -- ++(180:.5cm)  node[vertex,draw,solid] (v1){};
    \draw[densely dotted] (u1) -- ++(270:.5cm)  node[vertex,draw,solid] (w1){};
    \draw[gray,very thick,line width=2pt,opacity=.3] (v1) -- ++(135:.25cm)  node {};
    \draw[gray,very thick,line width=2pt,opacity=.3] (v1) -- ++(-135:.25cm)  node {};
    \draw[gray,very thick,line width=2pt,opacity=.3] (w1) -- ++(-45:.25cm)  node {};
    \draw[gray,very thick,line width=2pt,opacity=.3] (w1) -- ++(-135:.25cm)  node {};

    \draw[densely dotted] (C5below.corner 4) -- ++(-45:.5cm)  node[vertex,draw,label=south east:$x_k$,solid] (u2) {};
    \draw[densely dotted] (u2) -- ++(0:.5cm)  node[vertex,draw,solid] (v2){};
    \draw[densely dotted] (u2) -- ++(270:.5cm)  node[vertex,draw,solid] (w2){};
    \draw[gray,very thick,line width=2pt,opacity=.3] (v2) -- ++(45:.25cm)  node {};
    \draw[gray,very thick,line width=2pt,opacity=.3] (v2) -- ++(-45:.25cm)  node {};
    \draw[gray,very thick,line width=2pt,opacity=.3] (w2) -- ++(-45:.25cm)  node {};
    \draw[gray,very thick,line width=2pt,opacity=.3] (w2) -- ++(-135:.25cm)  node {};

    \foreach \x in {1,2,...,5} 
      \node[vertex,draw] at (C5below.corner \x) {};
}

\newcommand{\sprimedecomp}{$S'$-de\-com\-po\-si\-tion}
\newcommand{\clawdecomp}{$\{K_{1,3}\}$-de\-com\-po\-si\-tion}
\newcommand{\clawtriangledecomp}{$\{K_{1,3}$, $K_3\}$-de\-com\-po\-si\-tion}
\newcommand{\clawpathdecomp}{$\{K_{1,3}$, $P_4\}$-de\-com\-po\-si\-tion}
\newcommand{\clawtrianglepathdecomp}{$\{K_{1,3}$, $K_3$, $P_4\}$-de\-com\-po\-si\-tion}
\newcommand{\pathtriangledecomp}{$\{K_3$, $P_4\}$-de\-com\-po\-si\-tion}
\newcommand{\almostsubcubic}{degree-2,3}

\usepackage{enumitem} 

\newlist{aims}{itemize}{1}
\setlist[aims,1]{
  label={Question~\arabic*}
  ,leftmargin=*
  ,align=left
  ,   labelsep=0mm
}

\newcommand{\decisionproblem}[3]{
\begin{center}
\noindent\fbox{\parbox{33em}{
\begin{minipage}[t]{1\linewidth}
{\sc #1}
\begin{aims}
    \item[Input:] {#2}
    \item[Question:] {#3}
\end{aims}    
\end{minipage}
}
}    
\end{center}
}

\newtheorem{observation}{Observation}

\usepackage{cleveref}
\Crefname{observation}{Observation}{Observations}

\title{Decomposing Cubic Graphs into Connected Subgraphs of Size Three}
\author{
Laurent Bulteau\inst{1}
\and
Guillaume Fertin\inst{2}
\and 
Anthony Labarre\inst{1}
\and
Romeo Rizzi\inst{3}
\and
Irena Rusu\inst{2}}

\institute{
Université Paris-Est, LIGM (UMR 8049), CNRS, ENPC, ESIEE Paris, UPEM, F-77454, Marne-la-Vallée, France\\
\and Laboratoire d'Informatique de Nantes-Atlantique,
  UMR CNRS 6241, Universit\'e de Nantes, 2 rue de la Houssini\`ere, 44322
  Nantes Cedex 3, France\\
\and Department of Computer Science, University of Verona, Italy}

\begin{document}

\maketitle

\begin{abstract}
Let $S=\{K_{1,3},K_3,P_4\}$ be the set of connected graphs of size 3. We study 
the problem of partitioning the edge set of a graph $G$ into graphs taken from 
any non-empty $S'\subseteq S$. The problem is known to be \NP-complete for any 
possible choice of $S'$ in general graphs. In this paper, we assume that the 
input graph is cubic, and study the computational complexity of the problem of 
partitioning its edge set for any choice of $S'$. We identify all polynomial 
and \NP-complete problems in that setting, and give graph-theoretic 
characterisations of $S'$-decomposable cubic graphs in some cases.
\end{abstract}

\thispagestyle{plain}  
\pagestyle{plain} 

\section{Introduction} 

\paragraph{General context.}
Given a connected graph $G$ and a set $S$ of graphs, the {\sc $S$-decomposition}
problem asks whether $G$ can be represented as an edge-disjoint union
of subgraphs, each of which is isomorphic to a graph in $S$. The problem has a long
history that can be traced back to~\citet{Kirkman1847} and has been
intensively studied ever since, both from pure mathematical and
algorithmic point of views. One of the most notable results in the
area is the proof by~\citet{Dor1997} of the long-standing ``Holyer 
conjecture''~\cite{Holyer1981}, which stated that the {\sc
  $S$-decomposition} problem is \NP-complete when $S$ contains a
single graph with at least three edges. 

Many variants of the {\sc $S$-decomposition} problem have been studied
while attempting to prove Holyer's conjecture or to obtain
polynomial-time algorithms in restricted cases~\cite{Yuster200712},
and applications arise in such diverse fields as traffic
grooming~\cite{Munoz2011} and graph drawing~\cite{Fusy20091870}. 
In particular, \citet{Dyer1985} studied a variant where $S$ is the set of connected graphs
with $k$ edges for some natural $k$, and 
proved the \NP-completeness of the {\sc $S$-decomposition} problem for any $k\geq 3$, even under the assumption that
the input graph is planar and bipartite 
(see Theorem~3.1 in~\cite{Dyer1985}). They
further claimed 
that the problem remains
\NP-complete under the additional constraint that all vertices of the
input graph have degree either $2$ or $3$. Interestingly, if one looks
at the special case where $k=3$ and $G$ is a bipartite {\em cubic} graph (i.e., each vertex has degree
$3$), then $G$ can clearly be decomposed in polynomial time, using $K_{1,3}$'s only,
by selecting either part of the bipartition and making each vertex in that set the
center of a $K_{1,3}$. This shows that focusing on
the case $k=3$ and on cubic graphs can lead to tractable results --- 
as opposed to general graphs, for which when $k=3$, and for any non empty $S'\subseteq
S$, the {\sc $S'$-decomposition} problems all turn out to be
\NP-complete~\cite{Dyer1985,Holyer1981}. 

In this paper, we study the {\sc $S$-decomposition} problem on cubic graphs in the case $k=3$ --- i.e., $S=\{K_{1,3},K_3,P_4\}$. 
For any non-empty
$S'\subseteq S$, we settle the computational complexity of the
{\sc $S'$-decomposition} problem by showing that the problem is
\NP-complete when $S'=\{K_{1,3},P_4\}$ and $S'=S$, while all the
other cases are in \P.                          
\Cref{tab:summary-complexity-results} summarises the state of
knowledge regarding the complexity of decomposing cubic and arbitrary
graphs using connected subgraphs of size three, and puts our results
into perspective. 

\begin{table}[htbp]
\centering
\begin{tabular}{|c|c|c||c|c|}
    \hline
    \multicolumn{3}{|c||}{Allowed subgraphs} & \multicolumn{2}{c|}{Complexity according to graph class}\\
    \hline
    $K_{1,3}$ & $K_3$ & $P_4$ & cubic & arbitrary \\
    \hline
    \hline
    $\checkmark$ &  &  & in \P~(\Cref{prop:bipartite-cubic-graphs-are-easy-to-partition}) & \NP-complete~\cite[Theorem 3.5]{Dyer1985} \\

     & $\checkmark$ & & $O(1)$ (impossible) & \NP-complete~\cite{Holyer1981}\\

    &  & $\checkmark$ & in \P~\cite{Kotzig1957} & \NP-complete~\cite[Theorem 3.4]{Dyer1985} \\
    \hline
    $\checkmark$ & $\checkmark$ & & in
    \P~(\Cref{prop:cubic-graph-clawtriangle}) & 
\NP-complete
~\cite[Theorem 3.5]{Dyer1985}
\\
    $\checkmark$ & & $\checkmark$ &
    \NP-complete~(\Cref{corollary:clawpathdecomp-is-hard}) & \NP-complete~\cite[Theorem 3.1]{Dyer1985}\\

     & $\checkmark$ & $\checkmark$ & in \P~(\Cref{prop:cubic-graph-PM-iff-pathtriangle-partition}) & \NP-complete~\cite[Theorem 3.4]{Dyer1985}\\

    \hline
    $\checkmark$ & $\checkmark$ & $\checkmark$ & \NP-complete~(\Cref{corollary:main-problem-is-hard}) & \NP-complete~\cite[Theorem 3.1]{Dyer1985}\\
    \hline
\end{tabular} 
\caption{Known complexity results on decomposing graphs using subsets
  of $\{K_{1,3}, K_3, P_4\}$.}
\label{tab:summary-complexity-results}
\end{table}

\vspace*{-12mm}
\paragraph{Terminology.} We follow \citet{Brandstadt1987} for notation and terminology.
All graphs we consider are simple, connected and nontrivial
(i.e. $|V(G)|\geq 2$ and $|E(G)|\geq 1$).  
Given a set $S$ of graphs, a graph $G$ \emph{admits an
  $S$-decomposition}, or \emph{is $S$-decomposable}, if $E(G)$ can be
partitioned into subgraphs,  
each of which is isomorphic to a graph in $S$. Throughout the paper,
$S$ denotes the set of connected graphs of size 3, 
i.e. $S=\{K_3,K_{1,3},P_4\}$. We study the following problem:

\decisionproblem{{\sprimedecomp}}{a cubic graph $G=(V,E)$, a
non-empty set $S'\subseteq S$.}{
does $G$ admit a {\sprimedecomp}?
}

We let $G[U]$ denote the subgraph of $G$ induced by $U\subseteq V(G)$. Given a graph $G=(V,E)$, \emph{removing} a subgraph
$H=(V'\subseteq V,E'\subseteq E)$ of $G$ consists in 
removing edges $E'$ from $G$ as well as the possibly resulting
isolated vertices. Finally, 
let $G$ and $G'$ be two graphs. Then:
\begin{itemize} 
\item \emph{subdividing} an edge $\{u,v\}\in E(G)$ consists in inserting a
new vertex $w$ into that edge, so that $V(G)$ becomes $V(G)\cup \{w\}$
and $E(G)$ is replaced with $E(G)\setminus \{u,v\} \cup \{u,w\} \cup \{w,v\}$;  
\item \emph{attaching $G'$ to a vertex $u\in V(G)$} means building
 a new graph $H$ by identifying $u$ and some $v\in V(G')$; 
\item \emph{attaching $G'$ to an edge $e\in E(G)$} consists in
  subdividing $e$ using a new vertex $w$, then attaching $G'$ to $w$.
\end{itemize}

\Cref{fig:cubic-lobe} illustrates the process of attaching an edge to
an edge of the cube graph, and shows other small graphs that we will
occasionally use in this paper. 

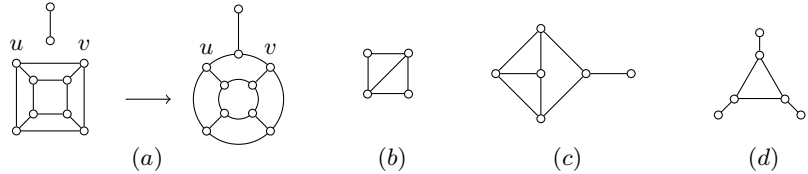
\begin{figure}[htbp]
\begin{center}
\setlength{\tabcolsep}{1.1em} 
\begin{tabular}{ccccc}
\adjustbox{valign=c}{
\begin{tikzpicture}[scale=.6]
    
    \node[vertex,draw] (a) at (0,0) {};
    \node[vertex,draw] (b) at (0,1.5) [label=above:$u$] {};
    \node[vertex,draw] (c) at (1.5,1.5)[label=above:$v$]  {};
    \node[vertex,draw] (d) at (1.5,0) {};
    \draw (a) -- (b) -- (c) -- (d) -- (a);
    \node[vertex,draw] (e) at (.375,1.125) {};
    \node[vertex,draw] (f) at (.375,.375) {};
    \node[vertex,draw] (g) at (1.125,1.125) {};
    \node[vertex,draw] (h) at (1.125,.375) {};
    \draw (e) -- (f) -- (h) -- (g) -- (e);
    \draw (a) -- (f);
    \draw (b) -- (e);
    \draw (c) -- (g);
    \draw (d) -- (h);

    \node[vertex,draw] (u) at (.75, 2) {};
    \node[vertex,draw] (v) at (.75, 2.75) {};
    \draw (u) -- (v);

    \begin{scope}[xshift=140pt,yshift=20pt]
        \draw[->] (-2.5, 0) -- (-1.5, 0);
        \trianglefreelobe
        \node at (b.south) [label=above:{$u$}] {};
        \node at (e.south) [label=above:{$v$}] {};
    \end{scope}
\end{tikzpicture}
}
&
\adjustbox{valign=c}{
\begin{tikzpicture}[scale=.6]
    \diamondgraph
\end{tikzpicture}
}
&
\adjustbox{valign=c}{
\begin{tikzpicture}[scale=.6]
    \cofish
\end{tikzpicture}
}
&
\adjustbox{valign=c}{
\begin{tikzpicture}[scale=.6]
    \netgraph
\end{tikzpicture}
}
\\
$(a)$ & $(b)$ & $(c)$ & $(d)$ 
\end{tabular}
\end{center}
\caption{$(a)$ Attaching a new edge to $\{u,v\}$; 
$(b)$ the diamond graph; $(c)$ the co-fish graph; $(d)$ the net graph.}
\label{fig:cubic-lobe}
\end{figure}

\section{Decompositions Without a $K_{1,3}$} 

In this section, we study decompositions of cubic graphs that use only $P_4$'s 
or $K_3$'s. 
Note that no cubic graph is $\{K_3\}$-decomposable, since all its vertices have odd degree. 
According to \citet{Bouchet1983131}, \citet{Kotzig1957} proved 
that a cubic graph admits a $\{P_4\}$-decomposition iff it
has a perfect matching. However, the proof of the forward direction as presented in~\cite{Bouchet1983131} is incomplete, as it requires the use of
\Cref{prop:p4k3-pas-3-chemins-incidents}.$(b)$ below, which is missing from their paper. 
Therefore, we provide the following proposition for completeness, together
with another result
which will
also be useful for the case where $S'=\{K_3,P_4\}$. 

\begin{proposition}
\label{prop:p4k3-pas-3-chemins-incidents}
Let $G$ be a cubic graph that admits a $\{K_3,P_4\}$-decomposition $D$. Then,
in $D$, (a)~no $K_3$ is used, and (b)~no three $P_4$'s are incident to the same vertex.
\end{proposition}

\begin{proof}
Partition $V(G)$ into three sets $V_1$, $V_2$ and $V_3$, where $V_1$ (resp. $V_2$, $V_3$) is the set of vertices that
are incident to exactly one $P_4$ (resp. two, three $P_4$'s) in $D$. Note that $V_1$
is exactly the set of vertices involved in $K_3$'s in $D$. Let $n_i=|V_i|$,
$1\leq i\leq 3$. Our goal is to show that $n_1=n_3=0$,
i.e. $V_1=V_3=\emptyset$. 
For this, note that (1)~each vertex in $V_3$ is the extremity of three different $P_4$'s, (2)~each
vertex in $V_2$ is simultaneously the extremity of one $P_4$ and
an inner vertex of another $P_4$, while (3)~each vertex in $V_1$ is
the extremity of one $P_4$. Since each $P_4$ has two extremities
and two inner vertices, if $p$ is the number of $P_4$'s in $D$, 
we have:
\begin{itemize}
\item[$\bullet$] $p=\frac{3n_3+n_2+n_1}{2}$ (by (1), (2) and (3) above, counting extremities);
\item[$\bullet$] $p=\frac{n_2}{2}$ (by (2) above, counting inner vertices).
\end{itemize}
Putting together the above two equalities yields $n_1=n_3=0$, which 
completes the proof.\qed
\end{proof}

Since $K_3$'s cannot be used in cubic graphs for $\{K_3,P_4\}$-decompositions by
\Cref{prop:p4k3-pas-3-chemins-incidents} above, we directly obtain the following result, which implies that {\pathtriangledecomp} is in \P.

\begin{proposition}\label{prop:cubic-graph-PM-iff-pathtriangle-partition}
A cubic graph admits a {\pathtriangledecomp} iff it has a perfect matching.
\end{proposition}

\section{Decompositions Without a $P_4$}\label{sec:clawtriangledecomp}

In this section, we study decompositions of cubic graphs that use only $K_{1,3}$'s 
or $K_3$'s. 

\begin{proposition}\label{prop:bipartite-cubic-graphs-are-easy-to-partition}
A cubic graph $G$ admits a $\{K_{1,3}\}$-decomposition iff it is bipartite.
\end{proposition}

\begin{proof}
For the reverse direction, select either set of the bipartition, and make each vertex in that set the center of a $K_{1,3}$. 
For the forward direction, let $D$ be a $\{K_{1,3}\}$-decomposition of $G$, and let $C$ and
$L$ be the sets of vertices containing, respectively, all the centers and all the
leaves of $K_{1,3}$'s in~$D$. We show that this is a
bipartition of $V(G)$. First, $C\cup L= V$ since $D$ covers all edges and therefore all vertices. Second, $C\cap L=\emptyset$ since a vertex in $C\cap L$ would have degree at least
4. Finally, each edge in $D$ connects the center of a $K_{1,3}$ and a leaf of
another $K_{1,3}$ in $D$, which belong respectively to $C$ and $L$. Therefore, 
$G$ is bipartite. 
\qed
\end{proof}

We now prove that $\{K_{1,3},K_3\}$-decompositions can be computed in
polynomial time. Recall that a graph is \emph{$H$-free} if it does not contain an
induced subgraph isomorphic to a given graph $H$. Since bipartite
graphs admit a $\{K_{1,3}\}$-decomposition
(by \Cref{prop:bipartite-cubic-graphs-are-easy-to-partition}), we can 
restrict our attention to non-bipartite graphs that contain $K_3$'s
(indeed, if they were $K_3$-free, then only $K_{1,3}$'s would be allowed and
\Cref{prop:bipartite-cubic-graphs-are-easy-to-partition} would imply that
they admit no decomposition). 
Our strategy consists in iteratively removing subgraphs from $G$
and adding them to an initially empty {\clawtriangledecomp} until $G$
is empty, in which case we have an actual decomposition, or no further
removal operations are possible, in which case no decomposition
exists.  Our analysis relies on the following notion: 
a $K_3$ induced by vertices $\{u, v, w\}$ in a graph $G$ is
\emph{isolated} if $V(G)$ contains no vertex $x$ such that
$\{u,v,x\}$,  $\{u,x,w\}$ or  $\{x,v,w\}$ induces a $K_3$.

\begin{lemma}\label{lemma:isolated-triangles-belong-to-decomposition}
 If a cubic graph $G$ admits a $\{K_{1,3},K_3\}$-decomposition $D$, then every isolated $K_3$ in $G$ belongs to $D$.
\end{lemma}

\begin{proof}[contradiction]
 If an isolated $K_3$ were not part of the decomposition, then
 exactly one vertex of that $K_3$ would be the center of a
 $K_{1,3}$, leaving the remaining edge uncovered and uncoverable. 
\qed
\end{proof}

$\overline{C_6}$ is a minimal example of a cubic
non-bipartite graph with $K_3$'s that admits no $\{K_{1,3},
K_3\}$-decomposition: both $K_3$'s in that graph must belong to the
decomposition
(by \Cref{lemma:isolated-triangles-belong-to-decomposition}), but their
removal yields a perfect matching.

\begin{observation}
\label{obs:removals-from-connected-cubic-graph-make-it-noncubic}
Let $G$ be a connected cubic graph. Then no sequence of at least one
edge or vertex removal from $G$ yields a cubic graph. 
\end{observation}

\begin{proof}[contradiction]
    If after applying at least one removal from $G$ we obtain
    a cubic graph $G'$, then the graph that precedes $G'$ in this removal
    sequence must have had a vertex of degree at least four, since $G$
    is connected.\qed 
\end{proof}

\begin{proposition}
\label{prop:all-triangles-isolated-then-clawtriangledecomp-poly}
    For any non-bipartite cubic graph $G$ whose $K_3$'s are all
    isolated, one can decide in polynomial time whether $G$ is
    $\{K_{1,3}, K_3\}$-decomposable. 
\end{proposition}

\begin{proof}
    We build a $\{K_{1,3}, K_3\}$-decomposition by iteratively
    removing $K_{1,3}$'s and $K_3$'s from $G$, which we add as we
    go to an initially empty set $D$.  
    By \Cref{lemma:isolated-triangles-belong-to-decomposition}, all
    isolated $K_3$'s must belong to $D$, so we start by adding them
    all to $D$ and removing them from $G$; therefore, $G$ admits a
    {\clawtriangledecomp} iff the resulting subcubic graph
    $G'$ admits a {\clawdecomp}.      
    Observe that $G'$ contains vertices of degree $1$ and $2$; we note that:
    \begin{enumerate}
        \item each vertex of degree $1$ must be the leaf of some $K_{1,3}$ in $D$;
        \item each vertex of degree $2$ must be the meeting point of two $K_{1,3}$'s in $D$.
    \end{enumerate}
    The only ambiguity arises for vertices of degree $3$, which may
    either be the center of a $K_{1,3}$ in $D$ or the meeting point of
    three $K_{1,3}$'s in $D$; however, there will always exist at
    least one other vertex of degree $1$ or $2$ until the graph is
    empty
    (by \Cref{obs:removals-from-connected-cubic-graph-make-it-noncubic}).  
    Therefore, we can safely remove $K_{1,3}$'s from our
    graph and add them to $D$ by following the above rules in the stated order; if we succeed in deleting the whole
    graph in this way, then $D$ is a {\clawtriangledecomp} 
    of $G$, otherwise no such decomposition exists.\qed 
\end{proof}

We conclude with the case where the graph may contain non-isolated $K_3$'s.

\begin{proposition}
    If a cubic graph $G$ contains a diamond, then one can decide in polynomial time whether $G$ is $\{K_{1,3},K_3\}$-decomposable.
\end{proposition}
\begin{proof}
    The only cubic graph on $4$ vertices is $K_4$, which is diamond-free and $\{K_{1,3},K_3\}$-decomposable, so we assume $|V(G)|\geq 6$.
Let $\mathcal{D}$ be a  diamond in $G$ induced by vertices 
    $\{u,v,w,x\}$ and such that $\{u,x\}\not\in
    E(G)$, as shown in 
    \Cref{fig:diamonds-and-clawtriangledecomps}$(a)$. $\mathcal{D}$ is connected to two other vertices $u'$ and $x'$ of $G$, which are respectively adjacent to $u$ and $x$, and there are only two ways to use the edges of $\mathcal{D}$ in a \clawtriangledecomp, as shown in
    \Cref{fig:diamonds-and-clawtriangledecomps}$(b)$ and $(c)$. If
    $u'=x'$, regardless of the decomposition we choose for
    $\mathcal{D}$, $u'$ and its neighbourhood induce a $P_3$ in the
    graph obtained from $G$ by removing the parts added to $D$. But then that $P_3$ cannot be covered, so no {\clawtriangledecomp}
    exists for $G$. Therefore, we assume that $u'\neq x'$. 

As \Cref{fig:diamonds-and-clawtriangledecomps}$(b)$ and $(c)$ show,
either $\{u,v,w\}$ or $\{v,w,x\}$ must form a $K_3$  in $D$, thereby
forcing either $\{v,w,x,x'\}$ or $\{u',u,v,w\}$ to form a $K_{1,3}$ in
$D$. In both cases, removing the $K_3$ and the $K_{1,3}$ yields a
graph $G'$ which contains vertices of degree $1$, $2$ or $3$. As in the 
proof of \Cref{prop:all-triangles-isolated-then-clawtriangledecomp-poly},
Observation~1 allows
us to make the following helpful observations: 
\begin{enumerate}
    \item every leaf in $G'$ must be the leaf of some $K_{1,3}$ in $D$;
    \item every vertex $y$ of degree two in $G'$ must either belong to a
      $K_3$ or be a leaf of two distinct $K_{1,3}$'s in $D$, which
      can be decided as follows:
\begin{enumerate}
    \item if $y$ belongs to a $K_3$ in $G'$, then it must also belong
      to a $K_3$ in $D$; otherwise, it would be the leaf of a
      $K_{1,3}$ and the graph obtained by removing that $K_{1,3}$
      would contain a $P_3$, which we cannot cover; 
    \item otherwise, $y$ must be a leaf of two $K_{1,3}$'s in $D$.
\end{enumerate}
\end{enumerate}

We therefore iteratively remove subgraphs from our graph and add
them to $D$ according to the above rules, which we follow in the
stated order; if we succeed in deleting the whole graph in this way
using either decomposition in
\Cref{fig:diamonds-and-clawtriangledecomps}$(b)$ or $(c)$ as a
starting point, then $D$ is a $\{K_{1,3}, K_3\}$-decomposition of $G$,
otherwise no such decomposition exists.\qed 
\end{proof}

\begin{figure}[htbp]
\setlength{\tabcolsep}{1.5em} 
\centering
\begin{tabular}{ccc}   
 \begin{tikzpicture}[scale=.4]
\diamondbeforepartition
\end{tikzpicture}
&
 \begin{tikzpicture}[scale=.4]
\diamondbeforepartition
\begin{scope}[every path/.style=part]
    \draw[Redish] (b) to (c) -- (d) to (b);
    \draw[Greenish] (b) to (a);
    \draw[Greenish] (a) to (d);
    \draw[Greenish] (a) to (6);
\end{scope}
\end{tikzpicture}
&
 \begin{tikzpicture}[scale=.4]
\diamondbeforepartition
\begin{scope}[every path/.style=part]
    \draw[Greenish] (b) to (c);
    \draw[Greenish] (c) to (d);
    \draw[Greenish] (5) to (c);
     \draw[Redish] (b) to (a) -- (d) to (b);
\end{scope}
\end{tikzpicture}
\\
$(a)$ & $(b)$ & $(c)$
\end{tabular}
\caption{$(a)$ A diamond in a cubic graph, and $(b), (c)$ the only two ways to decompose it in a {\clawtriangledecomp}.}
\label{fig:diamonds-and-clawtriangledecomps}
\end{figure}
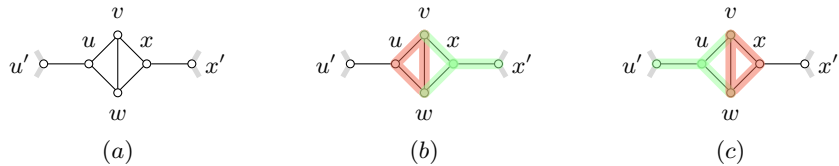

All the arguments developed in this section lead to the
following result.

\begin{proposition}
\label{prop:cubic-graph-clawtriangle}
The {\sc \clawtriangledecomp} problem on cubic graphs is in \P.
\end{proposition}

\section{Decompositions That Use Both $K_{1,3}$'s and $P_4$'s} 

\label{sec:hardness}

In this section, we show that problems {\sc \clawpathdecomp} and {\sc \clawtrianglepathdecomp} are \NP-complete. 
Our hardness proof relies on two intermediate problems that we define below and is structured as follows:

\scalebox{.75}{
\begin{minipage}{\textwidth}
\begin{align*}
        &\ \mbox{\sc cubic planar monotone 1-in-3 satisfiability} & \\
 \leq_P &\ \mbox{\sc degree-2,3 {\clawtrianglepathdecomp} with marked edges} & \mbox{(\Cref{thm:decomposition-with-marked-edges-is-hard} page~\pageref{thm:decomposition-with-marked-edges-is-hard})}\\
 \leq_P &\ \mbox{\sc {\clawtrianglepathdecomp} with marked edges} & \mbox{(\Cref{lemma:subcubic-marked-edges-reduces-to-cubic-marked-edges} page~\pageref{lemma:subcubic-marked-edges-reduces-to-cubic-marked-edges})}\\
 \leq_P &\ \mbox{\sc {\clawpathdecomp}} & \mbox{(\Cref{lemma:decomposition-with-marked-edges-reduces-to-main-problem} page~\pageref{lemma:decomposition-with-marked-edges-reduces-to-main-problem})}\\
\end{align*}
\end{minipage}
}

\noindent We start by introducing the following intermediate problem: 

\decisionproblem{{\clawtrianglepathdecomp} with marked edges}{a cubic
  graph $G=(V,E)$ and a subset $M\subseteq E$ of edges.}{does $G$
  admit a {\clawtrianglepathdecomp} $D$ such that no edge in $M$ is
  the middle edge of a $P_4$ in $D$ and such that every $K_3$ in
  $D$ has either one or two edges in $M$?} 

The drawings that illustrate our proofs in this section show
marked edges as dotted edges. The proof of
\Cref{lemma:decomposition-with-marked-edges-reduces-to-main-problem}
uses the following result. 

\begin{lemma}\label{lemma:bridge-must-be-middle-edge-of-P4}
Let $e$ be a bridge in a cubic graph $G$ which admits a $\{K_{1,3},K_3,P_4\}$-decomposition $D$. 
Then $e$ must be the middle edge of a $P_4$ in $D$.
\end{lemma}
\begin{proof}[contradiction]
First note that $e$ cannot belong to a $K_3$ in $D$. Now suppose $e$ is part of a $K_{1,3}$ in $D$. 
The situation is as shown below (without loss of generality):

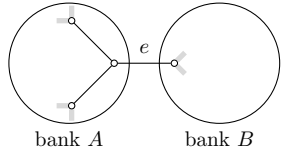
\begin{center}
    \begin{tikzpicture}[scale=.8,transform shape]
        \node[vertex,draw] at (0,0) (1) {};
        \node[vertex,draw] at (1,0) (2) {};
        
        \draw (2) -- (1)  -- ++(135:1cm) node[vertex,draw] (4) {};
        \draw (1)  -- ++(225:1cm) node[vertex,draw] (5) {};
        \draw (-.75, 0) circle (1cm);
        \draw (1.75, 0) circle (1cm);
        
        \node at (.5, .25) {$e$};
        \node at (-.75, -1.25) {bank $A$};
        \node at (1.75, -1.25) {bank $B$};
        
        \draw[gray,very thick,line width=2pt,opacity=.3] (5) -- ++(-90:.25cm)  node {};
        \draw[gray,very thick,line width=2pt,opacity=.3] (5) -- ++(180:.25cm)  node {};
        \draw[gray,very thick,line width=2pt,opacity=.3] (4) -- ++(90:.25cm)  node {};
        \draw[gray,very thick,line width=2pt,opacity=.3] (4) -- ++(180:.25cm)  node {};
        \draw[gray,very thick,line width=2pt,opacity=.3] (2) -- ++(45:.25cm)  node {};
        \draw[gray,very thick,line width=2pt,opacity=.3] (2) -- ++(-45:.25cm)  node {};
    \end{tikzpicture} 
\end{center}

If we remove from $G$ the $K_{1,3}$ in $D$ that contains $e$, 
then summing the terms of the degree sequence of $G[V(B)]$ yields  
$ 2 + 3(|V(B)|-1) = 2|E(B)|,$
which means that $2|E(B)|\equiv 2\pmod{3}$, so  $|E(B)|\not\equiv 0\pmod{3}$
and therefore $B$ admits no decomposition into components of size three. 
The very same argument 
shows that if $e$ belongs to a
$P_4$ in $D$, then it must be its middle edge, which
completes the proof.\qed 
\end{proof}

\begin{lemma}\label{lemma:decomposition-with-marked-edges-reduces-to-main-problem}
    Let $(G, M)$ be an instance of {\sc {\clawtrianglepathdecomp} with
      marked edges}, and $G'$ be the graph obtained by attaching a
    co-fish to every edge in $M$. Then $G$ can be decomposed iff $G'$ admits a 
\clawpathdecomp.
\end{lemma}
\begin{proof}
We prove each direction separately.
\begin{enumerate}
    \item [$\Rightarrow$:] we show how to transform a decomposition
      $D$ of $(G, M)$ into a decomposition $D'$ of $G'$. The subgraphs
      in $D$ that have no edge in $M$ are not modified. For the other
      subgraphs,  
    we distinguish between four cases:
    \begin{enumerate}
    \item if an edge of $M$ belongs to a $K_{1,3}$ in $D$, then attaching a co-fish does not prevent us from adapting the 
    decomposition of $G$ in $G'$: 

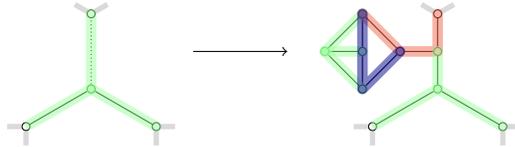
    \begin{center}
        \begin{tikzpicture}[scale=.5]
            \node[vertex,draw] (f) at (0,0) {};
            \draw[densely dotted] (f) -- ++(90:2cm) node [vertex,draw,solid] (1) {};
            \draw (f) -- ++(210:2cm) node [vertex,draw] (3) {};
            \draw (f) -- ++(330:2cm) node [vertex,draw] (4) {};

            \begin{scope}[every path/.style={part}]
                \draw [Greenish] (1) to (f) -- (3);
                \draw [Greenish] (f) to (4);
            \end{scope}

            \draw[gray,very thick,line width=2pt,opacity=.3] (1) -- ++(30:.5cm)  node {};
            \draw[gray,very thick,line width=2pt,opacity=.3] (1) -- ++(150:.5cm)  node {};

            \draw[gray,very thick,line width=2pt,opacity=.3] (3) -- ++(180:.5cm)  node {};
            \draw[gray,very thick,line width=2pt,opacity=.3] (3) -- ++(270:.5cm)  node {};

            \draw[gray,very thick,line width=2pt,opacity=.3] (4) -- ++(0:.5cm)  node {};
            \draw[gray,very thick,line width=2pt,opacity=.3] (4) -- ++(270:.5cm)  node {};
                
        \end{tikzpicture}
        \begin{tikzpicture}[scale=.5]
            \draw[->] (-4.5,0) -- (-2,0);
            \cofish;
            \draw (f) -- ++(90:1cm) node [vertex,draw] (1) {};
            \draw (f) -- ++(270:1cm) node [vertex,draw] (2) {};
            \draw (2) -- ++(210:2cm) node [vertex,draw] (3) {};
            \draw (2) -- ++(330:2cm) node [vertex,draw] (4) {};

            \begin{scope}[every path/.style={part}]
                
                \draw[Greenish] (b) to (c);
                \draw[Greenish] (c) to (e);
                \draw[Greenish] (c) to (d);
                \draw[Redish] (b) to (a) -- (f) to (1);
                \draw[Blueish] (b) to (e) -- (d) to (a);

                \draw [Greenish] (f) to (2) -- (3);
                \draw [Greenish] (2) to (4);
            \end{scope}

            \draw[gray,very thick,line width=2pt,opacity=.3] (1) -- ++(30:.5cm)  node {};
            \draw[gray,very thick,line width=2pt,opacity=.3] (1) -- ++(150:.5cm)  node {};

            \draw[gray,very thick,line width=2pt,opacity=.3] (3) -- ++(180:.5cm)  node {};
            \draw[gray,very thick,line width=2pt,opacity=.3] (3) -- ++(270:.5cm)  node {};

            \draw[gray,very thick,line width=2pt,opacity=.3] (4) -- ++(0:.5cm)  node {};
            \draw[gray,very thick,line width=2pt,opacity=.3] (4) -- ++(270:.5cm)  node {};
        \end{tikzpicture}
    \end{center}

    \item if  an edge of $M$ belongs to a $P_{4}$ in $D$, then it is an extremity of that $P_4$ and attaching a co-fish does not prevent us from adapting that part of the decomposition:

    
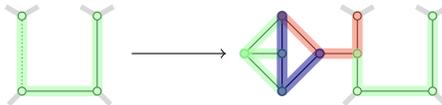
\begin{center}
    \begin{tikzpicture}[scale=.5]
        \node[vertex,draw] (f) at (0,0) {};
        \draw[densely dotted] (f) -- (270:2cm) node [vertex,draw,solid] (1) {} ;
        \draw (1)
                -- ++(0:2cm) node [vertex,draw] (2) {}
                -- ++(90:2cm) node [vertex,draw] (3) {};

        \draw[gray,very thick,line width=2pt,opacity=.3] (f) -- ++(30:.5cm)  node {};
        \draw[gray,very thick,line width=2pt,opacity=.3] (f) -- ++(150:.5cm)  node {};

        \draw[gray,very thick,line width=2pt,opacity=.3] (1) -- ++(225:.5cm)  node {};

        \draw[gray,very thick,line width=2pt,opacity=.3] (2) -- ++(-45:.5cm)  node {};

        \draw[gray,very thick,line width=2pt,opacity=.3] (3) -- ++(30:.5cm)  node {};
        \draw[gray,very thick,line width=2pt,opacity=.3] (3) -- ++(150:.5cm)  node {};
            \begin{scope}[every path/.style={part}]
                \draw [Greenish] (f) to (1) -- (2) to (3);
            \end{scope}

    \end{tikzpicture}
    \begin{tikzpicture}[scale=.5]
        \draw[->] (-4,0) -- (-1.5,0);
        \cofish;
        \draw (f) -- ++(90:1cm) node [vertex,draw] (4) {};
        \draw (f) -- ++(270:1cm) node [vertex,draw] (1) {} 
                -- ++(0:2cm) node [vertex,draw] (2) {}
                -- ++(90:2cm) node [vertex,draw] (3) {};
            
            \begin{scope}[every path/.style={part}]
                
                \draw[Greenish] (b) to (c);
                \draw[Greenish] (c) to (e);
                \draw[Greenish] (c) to (d);
                \draw[Redish] (b) to (a) -- (f) to (4);
                \draw[Blueish] (b) to (e) -- (d) to (a);

                \draw [Greenish] (f) to (1) -- (2) to (3);
            \end{scope}

        \draw[gray,very thick,line width=2pt,opacity=.3] (4) -- ++(30:.5cm)  node {};
        \draw[gray,very thick,line width=2pt,opacity=.3] (4) -- ++(150:.5cm)  node {};

        \draw[gray,very thick,line width=2pt,opacity=.3] (1) -- ++(225:.5cm)  node {};

        \draw[gray,very thick,line width=2pt,opacity=.3] (2) -- ++(-45:.5cm)  node {};

        \draw[gray,very thick,line width=2pt,opacity=.3] (3) -- ++(30:.5cm)  node {};
        \draw[gray,very thick,line width=2pt,opacity=.3] (3) -- ++(150:.5cm)  node {};
    \end{tikzpicture}
    \end{center}
    \item if a $K_3$ in $D$ has one edge in $M$, we can adapt the partition as follows:
    
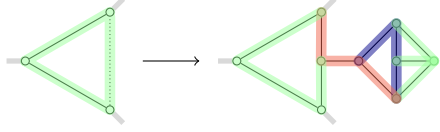
\begin{center}
        \begin{tikzpicture}[scale=.5]
            \path (180:1.5cm) node [vertex,draw] (1) {} ;
            \path (300:1.5cm) node [vertex,draw] (2) {};
            \path (60:1.5cm) node [vertex,draw] (3) {};
            \draw (1) -- (2);
            \draw[densely dotted] (2) -- (3);
            \draw (3) -- (1);
            \draw[gray,very thick,line width=2pt,opacity=.3] (1) -- ++(180:.5cm)  node {};
            \draw[gray,very thick,line width=2pt,opacity=.3] (2) -- ++(-45:.5cm)  node {};
            \draw[gray,very thick,line width=2pt,opacity=.3] (3) -- ++(45:.5cm)  node {};
            \begin{scope}[every path/.style={part}]
                \draw[Greenish] (1) to (2) -- (3) to (1);
            \end{scope}

        \begin{scope}[xshift=160pt]
            \path (180:1.5cm) node [vertex,draw] (1) {} ;
            \path (300:1.5cm) node [vertex,draw] (2) {};
            \path (60:1.5cm) node [vertex,draw] (3) {};
            \draw (1) -- (2) -- (3) -- (1);
            \draw[->] (-4,0) -- (-2.5,0);
            \draw[gray,very thick,line width=2pt,opacity=.3] (1) -- ++(180:.5cm)  node {};
            \draw[gray,very thick,line width=2pt,opacity=.3] (2) -- ++(-45:.5cm)  node {};
            \draw[gray,very thick,line width=2pt,opacity=.3] (3) -- ++(45:.5cm)  node {};
            \begin{scope}[rotate=180,xshift=-78pt]
                \cofish;
            
            \begin{scope}[every path/.style={part}]
                
                \draw[Blueish] (b) to (e) -- (d) to (a);
                \draw[Greenish] (c) to (b);
                \draw[Greenish] (c) to (e);
                \draw[Greenish] (c) to (d);
                \draw[Redish] (b) to (a) -- (f) to (3);
                
                \draw[Greenish] (3) to (1) -- (2) to (f);
            \end{scope}
            \end{scope}
        \end{scope}
        \end{tikzpicture}
    \end{center}
    \item if a $K_3$ in $D$ has two edges in $M$, we can adapt the partition as follows:
    
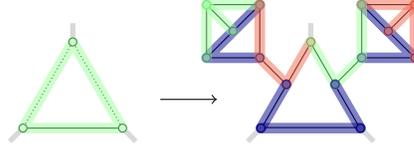
\begin{center}
        \begin{tikzpicture}[scale=.5]
            
            \node[minimum size=1.5cm,regular polygon,regular polygon sides=3] (square) {};

            \foreach \x/\lab in {1/a,2/b,3/c} 
                \node[vertex,draw] at (square.corner \x) (\x) {};

            \draw[densely dotted] (3) -- (1) -- (2);
            \draw (2) -- (3);

            \begin{scope}[every path/.style={part}]
                \draw[Greenish] (1) to (2) -- (3) to (1);
            \end{scope}

            \draw[gray,very thick,line width=2pt,opacity=.3] (1) -- ++(90:.5cm)  node {};
            \draw[gray,very thick,line width=2pt,opacity=.3] (3) -- ++(-45:.5cm)  node {};
            \draw[gray,very thick,line width=2pt,opacity=.3] (2) -- ++(-135:.5cm)  node {};

            \begin{scope}[xshift=180pt]

                \draw[->] (-4,0) -- (-2.5,0);

                \node[draw,minimum size=1.5cm,regular polygon,regular polygon sides=3] (square) {};

                \foreach \x/\lab in {1/a,2/b,3/c} 
                    \node[vertex,draw] at (square.corner \x) (\x) {};

                \draw[gray,very thick,line width=2pt,opacity=.3] (1) -- ++(90:.5cm)  node {};
                \draw[gray,very thick,line width=2pt,opacity=.3] (3) -- ++(-45:.5cm)  node {};
                \draw[gray,very thick,line width=2pt,opacity=.3] (2) -- ++(-135:.5cm)  node {};

                \begin{scope}[rotate=-45,xshift=-78pt,yshift=-5pt]
                    \cofish;
                    
                        \begin{scope}[every path/.style={part}]
                            \draw[Blueish] (b) to (e) -- (d) to (a);
                            \draw[Greenish] (b) to (c);
                            \draw[Greenish] (c) to (e);
                            \draw[Greenish] (c) to (d);
                            \draw[Redish] (b) to (a) -- (f) to (1);
                            \draw[Blueish] (f) to (2);
                        \end{scope}
                \end{scope}

                \begin{scope}[rotate=-135,xshift=-78pt,yshift=5pt]
                    \cofish;
                    
                        \begin{scope}[every path/.style={part}]
                            \draw[Blueish] (a) to (b) -- (d) to (e);
                            \draw[Redish] (b) to (c);
                            \draw[Redish] (c) to (e);
                            \draw[Redish] (c) to (d);
                            \draw[Greenish] (d) to (a) -- (f) to (1);
                            \draw[Blueish] (f) to (3);
                        \end{scope}
                \end{scope}

                \begin{scope}[every path/.style={part}]
                    \draw[Blueish] (3) to (2);
                \end{scope}
            \end{scope}
        \end{tikzpicture}
    \end{center}
    \end{enumerate}
    \item [$\Leftarrow$:] we now show how to transform any {\clawpathdecomp} $D'$ of $G'$ into a decomposition of $(G, M)$. Again, the only parts of $D'$ that will need adapting are those connected to the co-fishes that we inserted when transforming $G$ into $G'$. 
Since the leaf $u$ of the co-fish we inserted has a neighbour $x$ such that $\{u,x\}$ is a bridge in $G'$, $\{u,x\}$ is the middle edge of a $P_4$ in $D'$ (\Cref{lemma:bridge-must-be-middle-edge-of-P4}) and we may therefore assume without loss of generality that our starting point in $G'$ is as follows:
    
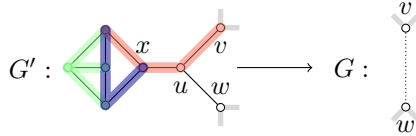
\begin{center}
        \begin{tikzpicture}[scale=.5]
            \begin{scope}
                \cofish;    
            \end{scope}
            \node at (f.north) [label=below:$u$] {};
            \node at (a.south) [label=above:$x$] {};
            \draw (f) -- ++( 45:1.5cm) node[vertex,draw,label=below:$v$] (v) {};
            \draw (f) -- ++(-45:1.5cm) node[vertex,draw,label=above:$w$] (w) {};

            \draw[gray,very thick,line width=2pt,opacity=.3] (v) -- ++(90:.5cm)  node {};
            \draw[gray,very thick,line width=2pt,opacity=.3] (v) -- ++(0:.5cm)  node {};
            \draw[gray,very thick,line width=2pt,opacity=.3] (w) -- ++(-90:.5cm)  node {};
            \draw[gray,very thick,line width=2pt,opacity=.3] (w) -- ++(0:.5cm)  node {};
            
            \begin{scope}[every path/.style={part}]
                \draw[Greenish] (c) to (b);
                \draw[Greenish] (c) to (e);
                \draw[Greenish] (c) to (d);
                \draw[Redish] (b) to (a) -- (f) to (v);
                \draw[Blueish] (a) to (d) -- (e) to (b);
            \end{scope}
                \node at (-2,0) {$G':$};
            \draw[->] (3.5,0) -- (5.5,0);
            \begin{scope}[xshift=300pt,rotate=-90]
                \node[vertex,draw,label=above:$v$] at (-1.05,-2.56) (v2) {};
                \node[vertex,draw,label=below:$w$] at (1.05,-2.56) (w2) {};

                \draw[gray,very thick,line width=2pt,opacity=.3] (v2) -- ++(135:.5cm)  node {};
                \draw[gray,very thick,line width=2pt,opacity=.3] (v2) -- ++(-135:.5cm)  node {};
                \draw[gray,very thick,line width=2pt,opacity=.3] (w2) -- ++(45:.5cm)  node {};
                \draw[gray,very thick,line width=2pt,opacity=.3] (w2) -- ++(-45:.5cm)  node {};
                \draw[densely dotted] (v2) -- (w2);
                \node at (0,-4) {$G:$};
            \end{scope}     
        \end{tikzpicture} 
    \end{center}
    with $\{v,w\}\not\in E(G')$ since $G$ is simple; therefore $\{u,w\}$ cannot belong to a $K_3$ in $G'$, and we have two cases to consider:
    \begin{enumerate}
        \item if $\{u, w\}$ belongs to a $K_{1,3}$ in $D'$, that $K_{1,3}$ can be mapped onto a $K_{1,3}$ in $D$ by replacing $\{u,w\}$ with $\{v,w\}$;
        \item otherwise, $\{u,w\}$ is an extremal edge of a $P_4$ in $D'$; since $\{u,w\}\not\in E(G)$, either that edge will remain in a $P_4$ when removing the co-fish and replacing $\{u,w\}$ with $\{v, w\}$, or it will end up in a $K_3$ with either one or two marked edges. Either way, the part can be added as such to $D$.\qed
    \end{enumerate}
\end{enumerate}
\end{proof}

We now show that we can restrict our attention to the following variant of {\sc {\clawtrianglepathdecomp} with marked edges}. We say a graph is \emph{{\almostsubcubic}} if its vertices have degree only 2 or 3.

\decisionproblem{{\almostsubcubic} {\clawtrianglepathdecomp} with marked edges}{a {\almostsubcubic} graph $G=(V,E)$ and a subset $M\subseteq E$ of edges.}{
does $G$ admit a {\clawtrianglepathdecomp} $D$ such that no edge
 in $M$ is the middle edge of a $P_4$ in $D$ and such that every $K_3$
 in $D$ has either one or two edges in $M$?
}

The following observation will help.

\begin{observation}\label{obs:subcubic-graph-num-leaves-deg-two}
    Let $G$ be a {\almostsubcubic} graph with $|V_2|$ degree-$2$ vertices. 
    If $G$ is $\{K_{1,3}, K_3, P_4\}$-decomposable, then $|V_2|\equiv 0\pmod{3}$.
\end{observation}
\begin{proof}
If $G=(V,E)$ admits a {\clawtrianglepathdecomp}, then $|E|\equiv 0\pmod{3}$. 
Let $V_2$ and $V_3$ be the subsets of vertices of degree $2$ and $3$ in $G$. Then $2|V_2|+3|V_3|=2|E|$, so $2|V_2|\equiv 0\pmod{3}$.
\qed
\end{proof}
We prove that allowing degree-2 vertices does not make the problem substantially more difficult, by adding the following gadgets until all vertices have degree 3. 
    Let $(G, M)$ be an instance of {\sc {\almostsubcubic} {\clawtrianglepathdecomp} with marked edges}, 
    where $G$ has at least three degree-2 vertices $t_1,t_2,t_3$; by 
    \emph{adding a net over $\{t_1,t_2,t_3\}$}, we mean attaching a net by its leaves to $v_1$, $v_2$ and $v_3$ and adding the edges incident to the net's leaves to $M$ (see \Cref{fig:algo-reduction}$(a)$).

\begin{figure}[tb]
\setlength{\tabcolsep}{1.5em} 
\centering
\begin{tabular}{cc}
\begin{tikzpicture}
    \path (90:1cm) node [vertex,draw,label=above:{$t_1$}] (a) {};
    \path (210:1cm) node [vertex,draw,label=south west:{$t_2$}] (b) {};
    \path (330:1cm) node [vertex,draw,label=south east:{$t_3$}] (c) {};

    \draw[gray,very thick,line width=2pt,opacity=.3] (a) -- ++(390:.25cm)  node {};
    \draw[gray,very thick,line width=2pt,opacity=.3] (a) -- ++(150:.25cm)  node {};

    \draw[gray,very thick,line width=2pt,opacity=.3] (b) -- ++(165:.25cm)  node {};
    \draw[gray,very thick,line width=2pt,opacity=.3] (b) -- ++(285:.25cm)  node {};

    \draw[gray,very thick,line width=2pt,opacity=.3] (c) -- ++(15:.25cm)  node {};
    \draw[gray,very thick,line width=2pt,opacity=.3] (c) -- ++(265:.25cm)  node {};

    \draw[densely dotted] (a) -- ++(270:.5cm)  node[vertex,draw,solid,label=right:$v_1$] (u) {};
    \draw[densely dotted] (b) -- ++(30:.5cm)  node[vertex,draw,solid,label=below:$v_2$] (v) {};
    \draw[densely dotted] (c) -- ++(150:.5cm)  node[vertex,draw,solid,label=below:$v_3$] (w) {};
    \draw (u) -- (v) -- (w) -- (u);
\end{tikzpicture}
&
\begin{tikzpicture}
    \path (90:1cm) node [vertex,draw,label=above:{$t_1$}] (a) {};
    \path (210:1cm) node [vertex,draw,label=south west:{$t_2$}] (b) {};
    \path (330:1cm) node [vertex,draw,label=south east:{$t_3$}] (c) {};

    \draw[gray,very thick,line width=2pt,opacity=.3] (a) -- ++(390:.25cm)  node {};
    \draw[gray,very thick,line width=2pt,opacity=.3] (a) -- ++(150:.25cm)  node {};

    \draw[gray,very thick,line width=2pt,opacity=.3] (b) -- ++(165:.25cm)  node {};
    \draw[gray,very thick,line width=2pt,opacity=.3] (b) -- ++(285:.25cm)  node {};

    \draw[gray,very thick,line width=2pt,opacity=.3] (c) -- ++(15:.25cm)  node {};
    \draw[gray,very thick,line width=2pt,opacity=.3] (c) -- ++(265:.25cm)  node {};

    \draw[densely dotted] (a) -- ++(270:.5cm)  node[vertex,draw,solid,label=right:$v_1$] (u) {};
    \draw[densely dotted] (b) -- ++(30:.5cm)  node[vertex,draw,solid,label=below:$v_2$] (v) {};
    \draw[densely dotted] (c) -- ++(150:.5cm)  node[vertex,draw,solid,label=below:$v_3$] (w) {};
    \draw (u) -- (v) -- (w) -- (u);
    
    \begin{scope}[every path/.style={Redish,line width=4pt,opacity=.5,cap=round,join=round}]
      \draw (a.center) -- (u.center) --(v.center);
     \draw (u.center) --(w.center);

    \end{scope}            
    \begin{scope}[every path/.style={Greenish,line width=4pt,opacity=.5,cap=round,join=round}]
         \draw (b.center) -- (v.center) --(w.center) --(c.center);

    \end{scope}            
    \begin{scope}[every path/.style={Redish,line width=4pt,opacity=.5,cap=round,join=round}]

    \end{scope} 
    
\end{tikzpicture}
\\
$(a)$ & $(b)$ \\
\end{tabular}
\caption{Adding a net $(a)$ to a graph with degree-2 vertices $t_1,t_2,t_3$ (dotted edges belong to $M'$), and $(b)$ its only possible decomposition (up to symmetry).}
\label{fig:algo-reduction}
\end{figure}
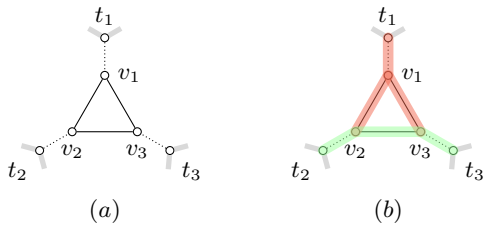

\begin{proposition}\label{prop:add-a-net}
    Let $(G, M)$ be an instance of {\sc {\almostsubcubic} {\clawtrianglepathdecomp} with marked edges}, 
    where $G$ has at least three degree-2 vertices $t_1,t_2,t_3$, and let  $(G', M')$ be the instance 
    obtained by adding a net to $(G, M)$.
     Then $(G', M')$ has three degree-2 vertices less than  $(G, M)$, and
     $(G, M)$ can be decomposed iff $(G', M')$ can be decomposed.
\end{proposition}
\begin{proof}
By construction, $G'$ has fewer degree-2 vertices, since $t_1,t_2,t_3$ now have degree 3 instead of 2, other vertices of $G$ are unchanged, and new vertices $\{v_1,v_2,v_3\}$ have degree 3. We now prove the equivalence.
\begin{enumerate}
    \item [$\Rightarrow$:] given a decomposition $D$ for $(G, M)$, we only need to add the $K_{1,3}$ induced by  $\{v_1,t_1,v_2,v_3\}$ and the $P_4$ induced by $\{t_2,v_2,v_3,t_3\}$ to cover the edges of the added net in order to obtain a decomposition $D'$ for $(G', M')$ (see \Cref{fig:algo-reduction}$(b)$).

\item [$\Leftarrow$:] 

    we show that the only valid decompositions 
must include the choice we made in the proof of the forward direction. Indeed,
    the marked edges cannot be middle edges in a $P_4$, and the $K_3$ induced by $v_1$, $v_2$ and $v_3$ cannot appear as a $K_3$ in a decomposition. Moreover, no marked edge can be the extremity of a $P_4$ with two edges lying in the $K_3$, since this would force another marked edge to be the middle edge of a $P_4$. Therefore the only possible decomposition of the net is the one defined above (up to symmetry), and we can safely remove the $P_4$ and the $K_{1,3}$ from $D'$ while preserving the rest of the decomposition.
\end{enumerate}
\qed
\end{proof}

\begin{lemma}\label{lemma:subcubic-marked-edges-reduces-to-cubic-marked-edges}
    {\sc {\almostsubcubic} {\clawtrianglepathdecomp} with marked edges} $\leq_P$     {\sc {\clawtrianglepathdecomp} with marked edges}.
\end{lemma}
\begin{proof}
Given an instance $(G, M)$ of {\sc {\almostsubcubic} {\clawtrianglepathdecomp} with marked edges},
create an instance $(G', M')$ by successively adding a net to any triple of degree-2 vertices, until no such triple remains. By \Cref{prop:add-a-net}, $(G, M)$ is decomposable iff $(G', M')$ is decomposable. Moreover, either $G'$ is cubic (hence $(G', M')$ is an instance of {\sc {\clawtrianglepathdecomp} with marked edges}), or $G$ is trivially a no-instance by Observation~\ref{obs:subcubic-graph-num-leaves-deg-two}
\qed
\end{proof}

Finally, we show that {\sc  {\almostsubcubic}  {\clawtrianglepathdecomp} with marked edges} is \NP-complete. 
Our reduction relies on the {\sc cubic planar monotone 1-in-3 satisfiability} 
problem~\cite{Moore:2001:HTP:2663034.2663094}:

\decisionproblem{cubic planar monotone 1-in-3 satisfiability}{
a Boolean formula $\phi=C_1\wedge C_2\wedge\cdots\wedge C_{n}$ without negations over a 
 set  $\Sigma = \{x_1, x_2, \ldots, x_m\}$, with exactly three distinct variables per 
 clause and where each literal appears in exactly three clauses; moreover, the 
graph with clauses and variables as 
vertices and edges joining clauses and the variables they contain is planar.
}{
does there exist an assignment of truth values $f : \Sigma \to \{\textsc{true}$, 
 $\textsc{false}\}$ such that exactly one literal is {\sc true} in every clause of $\phi$?}

\begin{theorem}\label{thm:decomposition-with-marked-edges-is-hard}
{\sc  {\almostsubcubic} {\clawtrianglepathdecomp} with marked ed\-ges} is \NP-com\-plete.
\end{theorem}
\begin{proof}
We first show how to transform an instance $\phi=C_1\wedge C_2\wedge\cdots\wedge C_{n}$ of {\sc cubic planar monotone 1-in-3 satisfiability} into an instance $(G, M)$ of  {\sc {\almostsubcubic} {\clawtrianglepathdecomp} with marked edges}. The transformation proceeds by:
\begin{enumerate}
    \item mapping each variable $x_i$ onto a $K_{1,3}$ denoted by $K(x_i)$ and whose edges all belong to $M$;
    \item mapping each clause $C=\{x_i,x_j,x_k\}$ onto a cycle with five vertices in such a way that $K(x_i)$, $K(x_j)$ and $K(x_k)$ each have a leaf that coincides with a vertex of the cycle and exactly two such leaves are adjacent in the cycle.
\end{enumerate}
\Cref{fig:reduction-from-cubic-mono-one-in-three-sat} illustrates the construction, which yields a  {\almostsubcubic} graph. 
We now show that $\phi$ is satisfiable iff $(G, M)$ admits a decomposition.
\begin{enumerate}
    \item [$\Rightarrow$:] we apply the following rules for transforming a satisfying assignment for $\phi$ into a decomposition $D$ for $(G, M)$:
    \begin{itemize}
        \item if variable $x_i$ is set to {\sc false}, then the corresponding $K(x_i)$ is added as such to $D$; 
        \item otherwise, the three edges of $K(x_i)$ will be the meeting points of three different $K_{1,3}$'s in the decomposition, one of which will have two edges in the current clause gadget.
    \end{itemize}
    Two cases can be distinguished based on whether or not a leaf of $K(x_i)$ is adjacent to a leaf of $K(x_j)$ or $K(x_k)$, but in both cases the rest of the clause gadget yields a $P_4$ that we add as such to the decomposition (see \Cref{fig:reduction-from-cubic-mono-one-in-three-sat}$(b)$ and $(c)$).

    \item [$\Leftarrow$:] we now show how to convert a decomposition $D$ for $(G, M)$ into a satisfying truth assignment for $\phi$. First, we observe that $D$ must satisfy the following crucial structural property:
\begin{center}
{\bf For each clause $C=(x_i\vee x_j\vee x_k)$, exactly two subgraphs out of $K(x_i)$, $K(x_j)$ and $K(x_k)$ appear as $K_{1,3}$'s in $D$}.
\end{center}

Indeed, $G$ is $K_3$-free by construction, and:
        \begin{enumerate}
            \item if all of them appear as $K_{1,3}$'s in $D$, then the remaining five edges of the clause gadget cannot be decomposed;
            \item if only $K(x_i)$ appear as a $K_{1,3}$ in $D$, then $x_j$ --- without loss of generality --- must be a leaf either of a $K_{1,3}$ in $D$ with a center in the clause gadget or of a $P_4$ in $D$ with two edges in the clause gadget (the $P_4$ cannot connect $x_j$ and $x_k$, otherwise the rest of the gadget cannot be decomposed); in both cases, the remaining three edges of the clause gadget must form a $P_4$, thereby causing $K(x_k)$ to appear as a $K_{1,3}$ in $D$, a contradiction (a similar argument allows us to handle $K(x_j)$ and $K(x_k)$);
            \item finally, if none of them appear as $K_{1,3}$'s in $D$, then $x_i$ must be the leaf either of a $K_{1,3}$ in $D$ with a center in the clause gadget, or of a $P_4$ with two edges in the clause gadget; in both cases, the remaining three edges of the clause gadget must form a $P_4$ in $D$, which in turn makes it impossible to decompose the rest of the graph.
        \end{enumerate}
    Therefore, $D$ yields a satisfying assignment for $\phi$ in the following simple way: if $K(x_i)$ appears as a $K_{1,3}$ in $D$, set it to {\sc false}, otherwise set it to {\sc true}.
\qed
\end{enumerate}
\end{proof}

\begin{theorem}\label{corollary:clawpathdecomp-is-hard}
    {\sc \clawpathdecomp} is \NP-complete.
\end{theorem}
\begin{proof}
    Immediate from \Cref{lemma:subcubic-marked-edges-reduces-to-cubic-marked-edges,lemma:decomposition-with-marked-edges-reduces-to-main-problem} and \Cref{thm:decomposition-with-marked-edges-is-hard}.\qed
\end{proof}

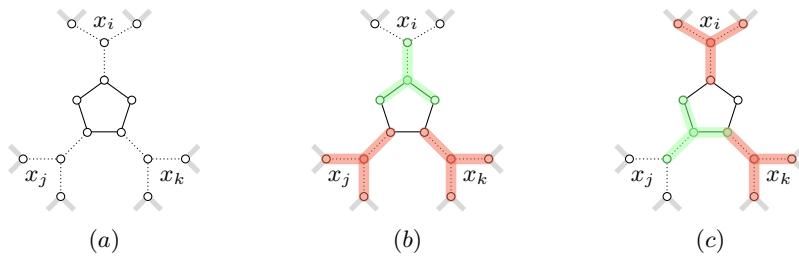
\begin{figure}[htbp]
\centering
\setlength{\tabcolsep}{2em} 
\begin{tabular}{ccc}
\begin{tikzpicture}
\varandclausegadget;
\end{tikzpicture}
&
\begin{tikzpicture}
    \varandclausegadget;
    \begin{scope}[every path/.style={Redish,line width=4pt,opacity=.5,cap=round,join=round}]
        \draw (C5below.corner 4) -- (u2.center) -- (v2.center);
        \draw (u2.center) -- (w2.center);   
    \end{scope}            
    \begin{scope}[every path/.style={Redish,line width=4pt,opacity=.5,cap=round,join=round}]
        \draw (C5below.corner 3) -- (u1.center) -- (v1.center);
        \draw (u1.center) -- (w1.center);   
    \end{scope}            
    \begin{scope}[every path/.style={Greenish,line width=4pt,opacity=.5,cap=round,join=round}]
        \draw (C5below.corner 1) -- (u.center);
        \draw (C5below.corner 2) -- (C5below.corner 1) -- (C5below.corner 5);   
    \end{scope}            
\end{tikzpicture}
&
\begin{tikzpicture}
    \varandclausegadget;
    \begin{scope}[every path/.style={Redish,line width=4pt,opacity=.5,cap=round,join=round}]
        \draw (C5below.corner 4) -- (u2.center) -- (v2.center);
        \draw (u2.center) -- (w2.center);   
    \end{scope}            
    \begin{scope}[every path/.style={Greenish,line width=4pt,opacity=.5,cap=round,join=round}]
        \draw (C5below.corner 3) -- (u1.center);
        \draw (C5below.corner 2) -- (C5below.corner 3) -- (C5below.corner 4);
    \end{scope}            
    \begin{scope}[every path/.style={Redish,line width=4pt,opacity=.5,cap=round,join=round}]
        \draw (C5below.corner 1) -- (u.center);
        \draw (v.center) -- (u.center) -- (w.center);   
    \end{scope}            
\end{tikzpicture}
\\
$(a)$ & $(b)$ & $(c)$ \\
\end{tabular}
\caption{$(a)$ 
Connecting clause and variable gadgets in the proof of \Cref{thm:decomposition-with-marked-edges-is-hard}; 
dotted edges belong to $M$. 
$(b), (c)$ Converting truth assignments into decompositions in the proof of  \Cref{thm:decomposition-with-marked-edges-is-hard}; the only variable set to {\sc true} is mapped onto a $K_{1,3}$ in the decomposition;  $(b)$ shows the case where the only variable set to {\sc true} --- namely, $x_i$ --- is such that $K(x_i)$ has no leaf adjacent to a leaf of $K(x_j)$ nor $K(x_k)$; $(c)$ shows the other case, where $x_j$ is set to {\sc true} and $K(x_i)$ and $K(x_k)$ have leaves made adjacent by the clause gadget.
}
\label{fig:reduction-from-cubic-mono-one-in-three-sat}
\end{figure}

A like-minded reduction\footnote[1]{See Appendix for details.} 
allows us to prove the hardness of {\sc \clawtrianglepathdecomp}.

\begin{theorem}\label{corollary:main-problem-is-hard}
        {\sc \clawtrianglepathdecomp} is \NP-complete, even on $K_3$-free graphs.
\end{theorem}

\section{Conclusions and Future Work}

We provided in this paper a complete complexity landscape of {\sc\clawtrianglepathdecomp} for cubic graphs. 
A natural generalisation, already studied by other authors, is to study decompositions of $k$-regular graphs into connected components with $k$ edges for $k>3$. We would like to determine whether our positive results generalise in any way in that setting. 
It would also be interesting to identify tractable classes of graphs in the cases where those decomposition problems are hard, and to refine our characterisation of hard instances; for instance, does there exist a planarity-preserving reduction for \Cref{corollary:main-problem-is-hard}? Finally, we note that some applications relax the size constraint by allowing the use of graphs with at most $k$ edges in the decomposition~\cite{Munoz2011}; we would like to know how that impacts the complexity of the problems we study in this paper. 

\vspace*{-3mm}
{
\bibliographystyle{mynatstyle}
\bibliography{partitioning-cubic-graphs}

\begin{thebibliography}{12}
\providecommand{\natexlab}[1]{#1}
\providecommand{\url}[1]{\texttt{#1}}
\expandafter\ifx\csname urlstyle\endcsname\relax
  \providecommand{\doi}[1]{doi: #1}\else
  \providecommand{\doi}{doi: \begingroup \urlstyle{rm}\Url}\fi

\bibitem[Bouchet and Fouquet(1983)]{Bouchet1983131}
{\sc A.~Bouchet and J.-L. Fouquet}, {\em Trois types de décompositions d'un
  graphe en chaînes}, in Combinatorial Mathematics: Proceedings of the
  International Colloquium on Graph Theory and Combinatorics, C.~Berge,
  D.~Bresson, P.~Camion, J.~F. Maurras, and F.~Sterboul, eds., vol.~75 of
  North-Holland Mathematics Studies, North-Holland, 1983, pp.~131--141.

\bibitem[Brandst\"{a}dt et~al.(1987)Brandst\"{a}dt, Le, and
  Spinrad]{Brandstadt1987}
{\sc A.~Brandst\"{a}dt, V.~B. Le, and J.~P. Spinrad}, {\em Graph classes: a
  survey}, {SIAM} Monographs on Discrete Mathematics and Applications, Society
  for Industrial Mathematics, 1987.

\bibitem[Dor and Tarsi(1997)]{Dor1997}
{\sc D.~Dor and M.~Tarsi}, {\em Graph decomposition is {NP}-complete: A
  complete proof of {H}olyer's conjecture}, SIAM J. Comput., 26 (1997),
  pp.~1166--1187.

\bibitem[Dyer and Frieze(1985)]{Dyer1985}
{\sc M.~E. Dyer and A.~M. Frieze}, {\em On the complexity of partitioning
  graphs into connected subgraphs}, Discrete Appl. Math., 10 (1985),
  pp.~139--153.

\bibitem[Fusy(2009)]{Fusy20091870}
{\sc {\'{E}}.~Fusy}, {\em Transversal structures on triangulations: {A}
  combinatorial study and straight-line drawings}, Discrete Math., 309 (2009),
  pp.~1870--1894.

\bibitem[Holyer(1981)]{Holyer1981}
{\sc I.~Holyer}, {\em {The NP-completeness of some edge-partition problems}},
  SIAM J. Comput., 10 (1981), pp.~713--717.

\bibitem[Kirkman(1847)]{Kirkman1847}
{\sc T.~P. Kirkman}, {\em On a problem in combinatorics}, Cambridge Dublin
  Mathematical Journal, 2 (1847), pp.~191--204.

\bibitem[Kotzig(1957)]{Kotzig1957}
{\sc A.~Kotzig}, {\em Z teorie konečných pravidelných grafov tretieho a
  štvrtého stupňa}, Časopis pro pěstování matematiky,  (1957),
  pp.~76--92.

\bibitem[Moore and Robson(2001)]{Moore:2001:HTP:2663034.2663094}
{\sc C.~Moore and J.~M. Robson}, {\em Hard tiling problems with simple tiles},
  Discrete Comput. Geom., 26 (2001), pp.~573--590.

\bibitem[Muñoz et~al.(2011)Muñoz, Li, and Sau]{Munoz2011}
{\sc X.~Muñoz, Z.~Li, and I.~Sau}, {\em Edge-partitioning regular graphs for
  ring traffic grooming with a priori placement of the {ADM}s}, SIAM J.
  Discrete Math., 25 (2011), pp.~1490--1505.

\bibitem[Schaefer(1978)]{DBLP:conf/stoc/Schaefer78}
{\sc T.~J. Schaefer}, {\em The complexity of satisfiability problems}, in Proc.
  10th STOC, San Diego, California, USA, May 1978, ACM, pp.~216--226.

\bibitem[Yuster(2007)]{Yuster200712}
{\sc R.~Yuster}, {\em Combinatorial and computational aspects of graph packing
  and graph decomposition}, Computer Science Review, 1 (2007), pp.~12--26.

\end{thebibliography}
}

\clearpage
\appendix\label{app:omitted-proofs}
\section{\appendixname: Omitted Proofs}
\label{sect:omitted_proofs}

\newcommand{\newDecomposition}{\ensuremath{K_3}-free {\clawpathdecomp} with marked edges}

Our hardness proof uses ideas similar to those used for {\clawpathdecomp}, and is based on a slightly different intermediate problem. The structure is as follows:

\newcommand{\mononaethreesat}{monotone not-all-equal $3$-satisfiability}

\scalebox{.8}{
\begin{minipage}{\textwidth}
\begin{align*}
        &\ \mbox{\sc monotone not-all-equal $3$-satisfiability} & \\
 \leq_P &\ \mbox{\sc \newDecomposition } & \mbox{(\Cref{thm:newDecomposition-with-marked-edges-is-hard} page~\pageref{thm:newDecomposition-with-marked-edges-is-hard})}\\
 \leq_P &\ \mbox{\sc {\clawtrianglepathdecomp}} & \mbox{(\Cref{lemma:K3-free-decomposition-with-marked-edges-reduces-to-main-problem} page~\pageref{lemma:K3-free-decomposition-with-marked-edges-reduces-to-main-problem})}\\
\end{align*}
\end{minipage}
}

\noindent We use the following intermediate problem.

\decisionproblem{\newDecomposition}{a cubic, $K_3$-free 
  graph $G=(V,E)$ and a subset $M\subseteq E$ of edges.}{does $G$
  admit a {\clawpathdecomp} $D$ such that no edge in $M$ is
  the middle edge of a $P_4$ in $D$?} 

\begin{lemma}\label{lemma:K3-free-decomposition-with-marked-edges-reduces-to-main-problem}
    Let $(G, M)$ be an instance of {\sc $K_3$-free {\clawpathdecomp} with
      marked edges}, and $G'$ be the graph obtained by attaching a
    co-fish to every edge in $M$. Then $G$ can be decomposed iff $G'$ admits a 
\clawtrianglepathdecomp.
\end{lemma}
\begin{proof}
The proof of the forward direction is exactly the same as that of the forward direction of \Cref{lemma:decomposition-with-marked-edges-reduces-to-main-problem}. For the reverse direction, let $D'$ be a {\clawtrianglepathdecomp} of $G'$. The only $K_3$'s in $G'$ are those that belong to the co-fishes we inserted, so we only need to show that removing those co-fishes does not prevent us from adapting the decomposition of $G'$ in order to obtain a $\{K_{1,3}, P_4\}$-decomposition $D$ of $G$. The proof is similar to that of the reverse direction of \Cref{lemma:decomposition-with-marked-edges-reduces-to-main-problem}, with the following modification: since $G$ is $K_3$-free, we have $N_G(u)\cap N_G(w)=\emptyset$, so if $\{u,w\}$ is the extremal edge of a $P_4$ in $D'$, then it will map onto $\{v,w\}$ in $G$, where it will become the extremal edge of a $P_4$ in $D$ (as opposed to, possibly, a $K_3$ in the proof of \Cref{lemma:decomposition-with-marked-edges-reduces-to-main-problem}).\qed
\end{proof}

We give a reduction from the following \NP-complete variant of {\sc sat}~\cite{DBLP:conf/stoc/Schaefer78}:

\decisionproblem{\mononaethreesat}{
a Boolean formula $\phi=C_1\wedge C_2\wedge\cdots\wedge C_{n}$ without negations over a 
 set  $\Sigma = \{x_1, x_2, \ldots, x_m\}$, with exactly three distinct variables per 
 clause.
}{
does there exist an assignment of truth values $f : \Sigma \to \{\textsc{true}$, 
 $\textsc{false}\}$ such that exactly one or two literals are {\sc true} in every clause of $\phi$?}

\begin{theorem}\label{thm:newDecomposition-with-marked-edges-is-hard}
{\sc  \newDecomposition} is \NP-com\-plete.
\end{theorem}
\begin{proof}
 Given an instance $\phi$ of {\sc \mononaethreesat}, we build an instance $(G,M)$ of {\sc \newDecomposition} as follows:
\begin{enumerate}
    \item 
 For each variable $x_i$ with $k$ occurrences (we can assume $k\geq 2$), we create a tree $T(x_i)$ with $3k$ leaves, called a \emph{variable tree}, whose edges are all marked (see \Cref{fig:varInNewDecomposition}$(a)$). Edges incident to a leaf are called \emph{border edges}, the others are called \emph{internal edges}. 
 
\item For each clause $x_i\vee x_j \vee x_k$, we create a $P_7$ $(v_1,v_2,\ldots, v_7)$, called a \emph{clause path}, with marked edges $\{v_2,v_3\}$ and $\{v_5,v_6\}$, to which we join three border edges of each variable tree as follows (see \Cref{fig:clauseInNewDecomposition}): 
\begin{itemize}
 \item[$\bullet$] one leaf of $T(x_1)$ is joined to $v_1$, another to $v_3$, and another to $v_7$;
 \item[$\bullet$] one leaf of $T(x_2)$ is joined to $v_1$, another to $v_5$, and another to $v_7$;
 \item[$\bullet$] one leaf of $T(x_2)$ is joined to $v_2$, another to $v_4$, and another to $v_6$.
 \end{itemize}
 \end{enumerate}

 Note that the resulting graph is indeed cubic: each inner vertex of a variable tree has degree 3, and each leaf is also part of a clause path. Furthermore, the inner vertices of a clause path are adjacent to two other vertices in the path and one other vertex in a variable tree, and each endpoint of a clause path is adjacent to another vertex in the path and two vertices in different variable trees. 
 
 We further observe that $G$ is $K_3$-free. First, any cycle included only in variable trees has length at least 4 (since it must be included in at least 2 such tree, and in each tree the path between any pair of leaves has length at least 2). Therefore, a $K_3$ would have to use 1 or 2 edges in a clause path. If it uses only one edge $\{v_i, v_{i+1}\}$, then both $v_i$ and $v_{i+1}$ must be joined to leaves of the same variable tree, which is impossible since each clause consists of three different variables. If two edges are used, then the last edge of the $K_3$ would be joining two leaves in some variable tree, which is also impossible. Therefore, $G$ is $K_3$-free.

\begin{figure}[tb]
\newcommand{\varTree}{

    \foreach \x [count=\i] in {-1,1,4}{
      \node[vertex,draw] (a\i) at (\x,2) {};
    }
    \foreach \x [count=\i] in {-2,-1,0,1,2,4,5}{
      \node[vertex,draw] (b\i) at (\x,1) {};
    }
    \foreach \x [count=\i] in {-2.25,-1.75,-1.25,-0.75,0,0.75,1.25,2,2,3.75,4.25,4.75,5.25}{
      \node[vertex,draw] (c\i) at (\x,0) {};
    }
    \node (x1) at (2.4,1.4) {};
    \node (x2) at (3.6,1.6) {};
    
    \foreach \i/\j in {1/1,1/2,1/3,2/3,2/4,2/5,3/6,3/7} {
      \draw[densely dotted](a\i)--(b\j);
    }
    
    \draw[densely dotted] (b5)--(x1);
    \draw[densely dotted] (a3)--(x2);
    \node at (3,1.5){\ldots};
    
    \foreach \i/\j in {1/1,1/2,2/3,2/4,3/5,4/6,4/7,5/8,6/10,6/11,7/12,7/13} {
      \draw[densely dotted](b\i)--(c\j);
    }
}

\centering
\begin{tikzpicture}
    \draw[cap=round, line width=5mm, gray!20] (-2.5,2) -- (7.5,2);
    \draw[cap=round, line width=5mm, gray!20] (-2.5,1) -- (7.5,1);
    \draw[cap=round, line width=5mm, gray!20] (-2.5,0) -- (7.5,0);
    \node at (7,2) {$A$};
    \node at (7,1) {$B$};
    \node at (7,0) {leaves};
    \varTree
\end{tikzpicture}
\\
$(a)$
\\
\begin{tikzpicture}
 
    \varTree
    
    \begin{scope}[every path/.style=part]
    \draw[red] (a1) to (b1);
    \draw[red] (b1) to (c1);
    \draw[red] (b1) to (c2);
    \draw[green] (a1) to (b2);
    \draw[green] (b2) to (c3);
    \draw[green] (b2) to (c4);
    \draw[blue] (a1) to (b3);
    \draw[blue] (b3) to (a2);
    \draw[blue] (b3) to (c5);
    \draw[green] (a2) to (b4);
    \draw[green] (b4) to (c6);
    \draw[green] (b4) to (c7);
    \draw[red] (a2) to (b5);
    \draw[red] (b5) to (c8);
    \draw[red] (b5) to (x1);    
    \draw[red] (x2) to (a3);
    
    \draw[green] (a3) to (b6);
    \draw[green] (b6) to (c10);
    \draw[green] (b6) to (c11);
    \draw[blue] (a3) to (b7);
    \draw[blue] (b7) to (c12);
    \draw[blue] (b7) to (c13);

    \end{scope}
\end{tikzpicture}
\\ $(b)$ \\
\begin{tikzpicture}

    \varTree
    \begin{scope}[every path/.style=part]
    \draw[red] (a1) to (b1);
    \draw[red] (a1) to (b2);
    \draw[red] (a1) to (b3);
    \draw[green] (a2) to (b3);
    \draw[green] (a2) to (b4);
    \draw[green] (a2) to (b5);
    \draw[blue] (b5) to (x1);
    \draw[blue] (a3) to (x2);
    \draw[blue] (a3) to (b6);
    \draw[blue] (a3) to (b7);
    \end{scope}

\end{tikzpicture}
\\
$(c)$
\caption{$(a)$ A variable tree $T(x_i)$ for a variable $x_i$ with $k$ occurrences: all its edges are marked, and it has $3k$ leaves. Internal vertices are partitioned into two sets $A$ and $B$.
$(b)$ A decomposition of $T(x_i)$, corresponding to $x_i=\mbox{{\sc true}}$. $(c)$ A decomposition of the internal edges of $T(x_i)$, corresponding to $x_i=\mbox{{\sc false}}$. }
\label{fig:varInNewDecomposition}
\end{figure}
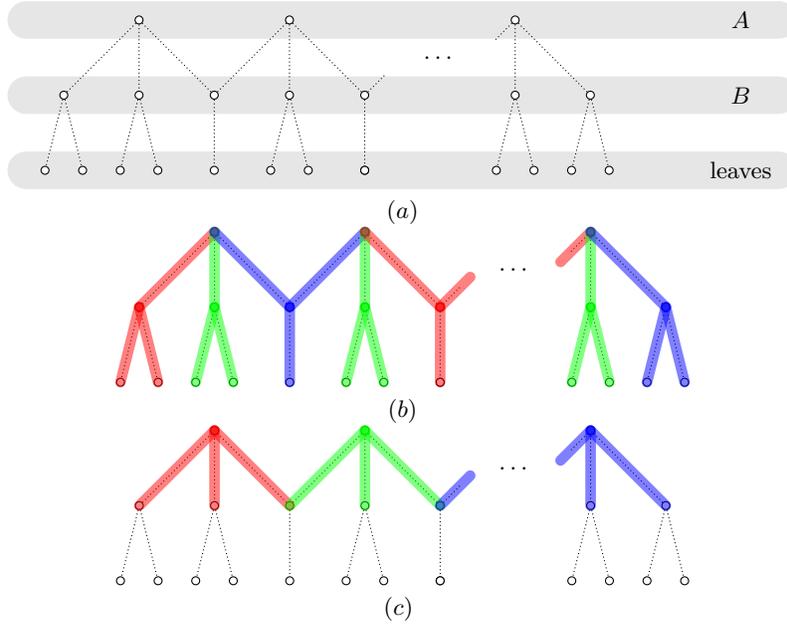

\begin{figure}
\centering
\newcommand{\clausePath}{
    \foreach \x/\y [count=\i] in {-3/1,-2/0,-1/1,0/0,1/1,2/0,3/1}{
      \node[vertex,draw, label=below:$v_\i$] (a\i) at (\x,\y) {};
    }
    \draw (a1)--(a2);
    \draw[densely dotted] (a2)--(a3);
    \draw (a3)--(a4);
    \draw (a4)--(a5);
    \draw[densely dotted] (a5)--(a6);
    \draw (a6)--(a7);
    
    }
    
\newcommand{\coverableTs}{
    \node[circle,fill=white] (T1) at (-2,3) {$T(x_i)$};
    \node[circle,fill=white] (T2) at (2,3) {$T(x_j)$};
    \node[circle,fill=white] (T3) at (0,-1.5) {$T(x_k)$};
    \foreach \a/\t in {1/1,1/2,2/3,3/1,4/3,5/2,6/3,7/1,7/2} {
      \draw[densely dotted] (T\t) -- (a\a);
    }
    
 }
    
\begin{tikzpicture}
\clausePath
    \node[circle, fill=gray!30] (T1) at (-2,3) {$T(x_i)$};
    \node[circle, fill=gray!30] (T2) at (2,3) {$T(x_j)$};
    \node[circle, fill=gray!30] (T3) at (0,-1.5) {$T(x_k)$};
    \foreach \a/\t in {1/1,1/2,2/3,3/1,4/3,5/2,6/3,7/1,7/2} {
      \draw[gray,very thick,line width=2pt,opacity=.3] (T\t) -- (a\a);
    }
\end{tikzpicture}
\\$(a)$
\begin{tabular}{cc}

\begin{tikzpicture}
\clausePath
\coverableTs

    \begin{scope}[every path/.style=part]
    \draw[red] (a2) to (a1);
    \draw[red] (a2) to (a3);
    \draw[red] (a2) to (T3);
    \draw[green] (a4) to (a3);
    \draw[green] (a4) to (a5);
    \draw[green] (a4) to (T3);    
    \draw[blue] (a6) to (a5);
    \draw[blue] (a6) to (a7);
    \draw[blue] (a6) to (T3);
    \end{scope}
\coverableTs
\end{tikzpicture}
&
\begin{tikzpicture}
\clausePath
\coverableTs
    \begin{scope}[every path/.style=part]
    \draw[red] (a2) to (a1);
    \draw[red] (a2) to (a3);
    \draw[red] (a1) to (T2);
    \draw[green] (a4) to (a3);
    \draw[green] (a4) to (a5);
    \draw[green] (a5) to (T2);    
    \draw[blue] (a6) to (a5);
    \draw[blue] (a6) to (a7);
    \draw[blue] (a7) to (T2);
    \end{scope}

\coverableTs
\end{tikzpicture}
\\
$(b)$ $x_i=\mbox{\sc true}$, $x_j=\mbox{\sc true}$, $x_k=\mbox{\sc false}$
&
$(c)$ $x_i=\mbox{\sc true}$, $x_j=\mbox{\sc false}$, $x_k=\mbox{\sc true}$\\
\begin{tikzpicture}
\clausePath
\coverableTs
    \begin{scope}[every path/.style=part]
    \draw[red] (T1) to (a1);
    \draw[red] (a1) to (a2);
    \draw[red] (a2) to (T3);
    \draw[green] (a3) to (a2);
    \draw[green] (a3) to (a4);
    \draw[green] (a3) to (T1);    
    \draw[blue] (T3) to (a4);
    \draw[blue] (a4) to (a5);
    \draw[blue] (a5) to (a6);
    \draw[green] (T3) to (a6);
    \draw[green] (a6) to (a7);
    \draw[green] (a7) to (T1);        
    \end{scope}
    \coverableTs
\end{tikzpicture}
&
\begin{tikzpicture}
\clausePath
\coverableTs
    \begin{scope}[every path/.style=part]
    \draw[red] (a2) to (a1);
    \draw[red] (a1) to (T1);
    \draw[red] (a1) to (T2);
    \draw[green] (a3) to (a2);
    \draw[green] (a3) to (T1);
    \draw[green] (a3) to (a4);    
    \draw[blue] (a5) to (a4);
    \draw[blue] (a5) to (a6);
    \draw[blue] (a5) to (T2);
    \draw[green] (a7) to (a6);
    \draw[green] (a7) to (T1);
    \draw[green] (a7) to (T2);    
    \end{scope}
    \coverableTs
\end{tikzpicture}
\\
$(d)$ $x_i=\mbox{\sc false}$, $x_j=\mbox{\sc true}$, $x_k=\mbox{\sc false}$
& 
$(e)$ $x_i=\mbox{\sc false}$, $x_j=\mbox{\sc false}$, $x_k=\mbox{\sc true}$
\end{tabular}

\caption{$(a)$ A clause path and its connections to the variable trees. $(b-e)$ For each truth assignment of $x_i$, $x_j$, $x_k$ (up to symmetry), a decomposition of the path edges and the neighboring uncovered edges of variable trees for false variables.
}
\label{fig:clauseInNewDecomposition}
\end{figure}
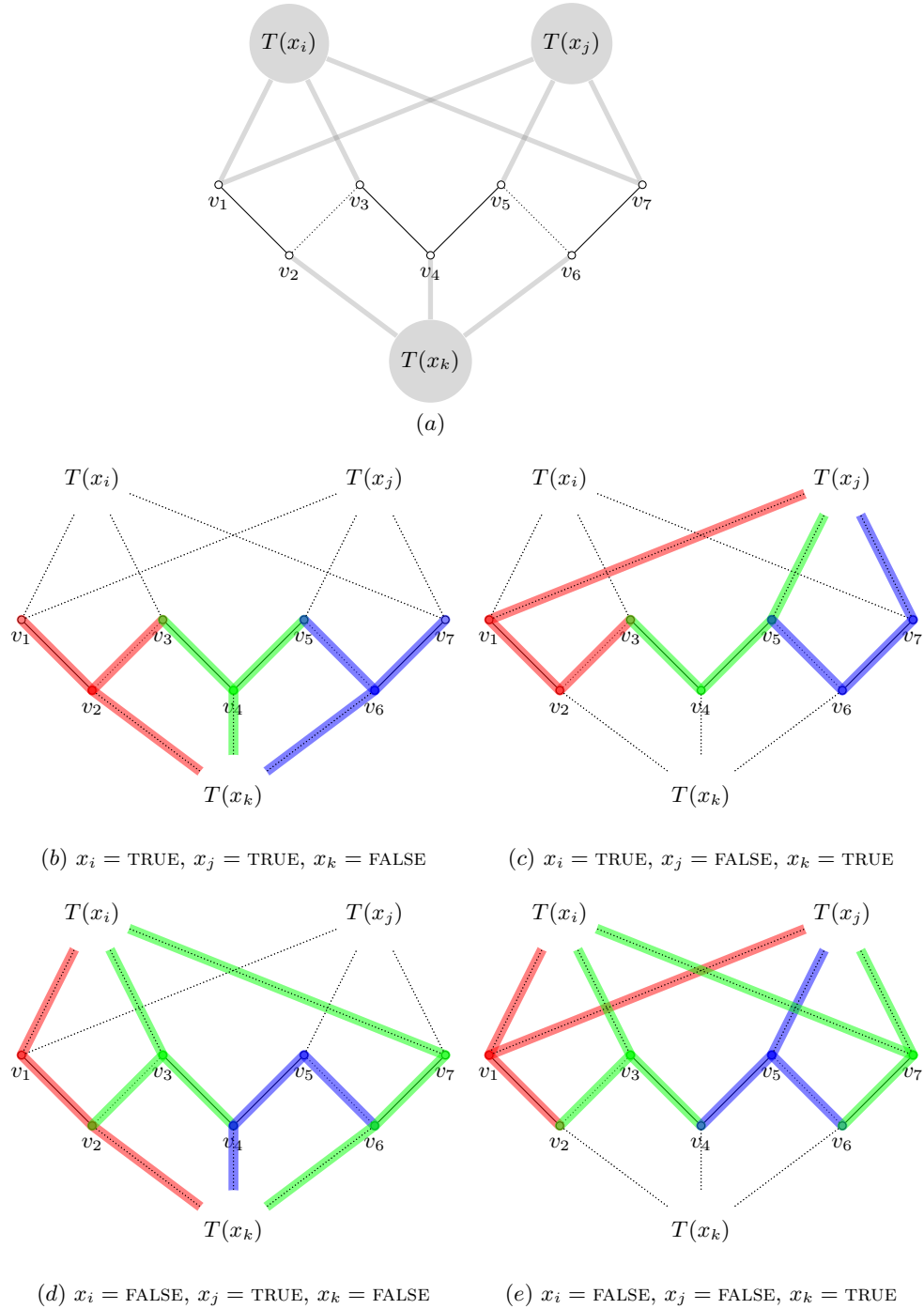
 
 We now prove that $\phi$ is satisfiable iff $(G,M)$ admits a decomposition.
 \begin{enumerate}
 \item [$\Rightarrow$:]
    From a satisfying truth assignment, we create a {\clawpathdecomp} of $(G,M)$ as described in Figures~\ref{fig:varInNewDecomposition}$(b-c)$ and~\ref{fig:clauseInNewDecomposition}$(b-e)$.
    Specifically, for each variable $x_i$ with $k$ occurrences, if $x_i=\mbox{\sc true}$, then we cover all edges of $T(x_i)$ with $2k-1$ $K_{1,3}$'s (Figure~\ref{fig:varInNewDecomposition}$(b)$). If $x_i=\mbox{\sc false}$, then we cover all internal edges of $T(x_i)$ with $k-1$ $K_{1,3}$'s(Figure~\ref{fig:varInNewDecomposition}$(c)$). 
    
    Now for any clause $x_i\vee x_j\vee x_k$, at least 1 and at most 2 variables among $\{x_i,x_j,x_k\}$ are set to {\sc false}. For these two variables and the corresponding variable trees, the border edges are still uncovered (for variables assigned {\sc true}, border edges are covered with the rest of the variable tree). Each tree has 3 border edges coming to the clause path, so there are either 9 or 12 edges to cover (6 in the clause path and 3 or 6 in border variable trees). As shown in Figures~\ref{fig:varInNewDecomposition}$(b-e)$, there always exist a decomposition of these edges into $K_{1,3}$'s and $P_4$'s (with the constraint that no marked edge is the middle of a $P_4$).    
    
    Overall, all edges in all variable trees and clause paths are covered by a $K_{1,3}$ or a $P_4$, therefore $(G,M)$ admits a {\clawpathdecomp}.    
    
 \item [$\Leftarrow$:]
   We first consider any variable tree, and show that the decompositions used for {\sc true} and {\sc false} assignments above are in fact the only two possible decompositions of the internal edges of this tree. Indeed, consider internal vertices and their partition into $A$ and $B$ (see \Cref{fig:varInNewDecomposition}$(a)$): because of marked edges, all edges between any two internal vertices must be part of a $K_{1,3}$, linking a leaf to a center. Since internal vertices form a path, they must alternate along this path between leaves and centers of $K_{1,3}$'s, so either all vertices of $B$ are centers, or all vertices of $A$ are centers. Each case yields only one possible decomposition of adjacent edges, as described respectively in Figures~\ref{fig:varInNewDecomposition}$(b)$ and \ref{fig:varInNewDecomposition}$(c)$. Naturally, we assign {\sc true} to any variable whose tree is decomposed as in the first case, and {\sc false} to other variables. We further make the following observation for the {\sc false} case: consider any leaf $x$ of the tree, and its parent $y$. Due to marked edges, $x$ can only be the center of a $K_{1,3}$, or a middle node of a $P_4$, of which $y$ is an endpoint.
   
   We now consider a clause $x_i\vee x_j\vee x_k$, and show that its variables can neither be all set to {\sc true} nor all set to  {\sc false}. Aiming at a contradiction, assume first that all variables are set to {\sc true}. Then all border edges of their trees are already covered, and the $K_{1,3}$'s and $P_4$'s covering the clause path may only use the 6 edges of the path. The only possibility to decompose the path in such a way is to use two $P_4$'s, however, such $P_4$'s would have marked edges as middle edges, which is forbidden.  We now assume that all variables of a clause are set to {\sc false}, i.e. it remains to cover the clause path and all border edges of the variable trees. Thanks to the  observation made on leaves of the tree in the {\sc false} case, vertices $v_1$ and $v_2$ are either centers of $K_{1,3}$'s, or middle nodes in $P_4$'s. They cannot both be centers of $K_{1,3}$'s, and due to marked edges, they must both be middle nodes of the same $P_4$. However, as noted above, the parents of $v_1$ in both trees $T(x_i)$ and $T(x_j)$ should be endpoints of this $P_4$, which is impossible.
   
   Finally, each variable has been assigned a truth value, and for each clause there must be at least one {\sc true} and one {\sc false} variable: therefore, we have a satisfying assignment for our instance of {\sc \mononaethreesat}.

 \end{enumerate}
 \qed
\end{proof}

\end{document}